\documentclass[english,12pt]{article} 
\usepackage{epsfig,amsfonts,amssymb,amsmath,amsthm}
\usepackage[T1]{fontenc}
\usepackage[latin9]{inputenc}
\usepackage{verbatim}
\usepackage{graphicx}
\usepackage{esint}

\usepackage{babel}
\title{Arbitrage Opportunities in Misspecified Stochastic Volatility Models}
\author{Rudra P. Jena\footnote{Centre de Mathématiques Appliquées, Ecole Polytechnique, 91128 Palaiseau France. E-mail: {\tt \{jena,tankov\}@cmap.polytechnique.fr}}
 \and Peter Tankov\footnotemark[1]\ \footnote{Corresponding author}}
\date{}
\begin{document}
\newtheorem{theo}{Theorem}[section]
\newtheorem{lemma}[theo]{Lemma}
\newtheorem{propo}[theo]{Proposition}
\newtheorem{coro}[theo]{Corollary}

\theoremstyle{remark}

\newtheorem{defi}[theo]{Definition}
\newtheorem{notation}[theo]{Notation}
\newtheorem{remark}[theo]{Remark}
\newtheorem{example}[theo]{Example}
\newtheorem{cond}[theo]{Condition}
\newtheorem{assu}[theo]{Assumption}
\newtheorem{question}[theo]{Question}
\newtheorem{conclusion}[theo]{Conclusion}


\voffset -0.7 true cm \hoffset 1.1 true cm \topmargin 0.0in
\evensidemargin 0.0in \oddsidemargin 0.0in \textheight 8.6in
\textwidth 7.25in \parskip 10 pt


\maketitle

  \begin{abstract}
There is vast empirical evidence that given a set of assumptions on the real-world dynamics of an asset, the European options on this asset are not efficiently priced in options markets, giving rise to arbitrage opportunities. We study these opportunities in a generic stochastic volatility model and exhibit the strategies which maximize the arbitrage profit. In the case when the misspecified dynamics is a classical Black-Scholes one, we give a new interpretation of the butterfly and risk reversal contracts in terms of their performance for volatility arbitrage. Our results are illustrated by a numerical example including transaction costs. 
\end{abstract}




\section{Introduction}
It has been observed by several authors \cite{awy01,bakshi,dumas} that given a set of assumptions on the real-world dynamics of the underlying, the European options on this underlying are not efficiently priced in options markets. Important discrepancies between the implied volatility and historical volatility levels, as illustrated in Figure \ref{revol}, as well as substantial differences between historical and option-based measures of skewness and kurtosis \cite{bakshi} have been documented. These discrepancies could be explained by systematic mispricings / model misspecification in option markets, leading to potential arbitrage opportunities\footnote{There exist many alternative explanations for why the implied volatilities are consistently higher than historical volatilities, such as, price discontinuity \cite{bates}, market crash fears \cite{bates00} and liquidity effects such as transaction costs \cite{leland.85,kabanov.safarian.97,cetin.al.06,cetin.soner.touzi.10}, inability to trade in continuous time \cite{bertsimas.kogan.lo.00} and market microstructure effects \cite{rr3}. The literature is too vast to cite even the principal contributions here}.
 The aim of this paper is to quantify these opportunities within a generic stochastic volatility framework, and to construct the strategies maximizing the gain. The arbitrage opportunities analyzed in this paper can be called statistical arbitrage opportunities, because their presence / absence depends on the statistical model for the dynamics of the underlying asset. Contrary to model independent arbitrages, such as violation of the call-put parity, a statistical arbitrage only exists in relation to the particular pricing model.

\begin{figure}
\centerline{\includegraphics[width=0.6\textwidth]{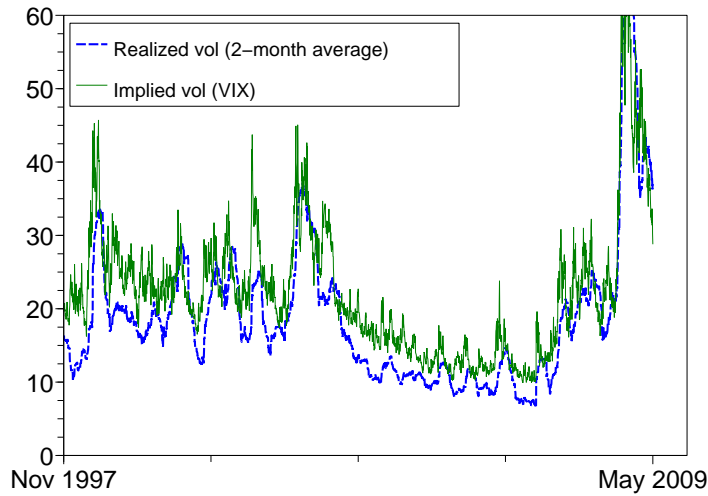}}
\caption{Historical evolution of the VIX index (implied volatility of options on the S\&P 500 index with 1 to 2 month to expiry, averaged over strikes, see \cite{vixwhitepaper}) compared to the historical volatility of the S\&P 500 index. The fact that the implied volatility is consistently above the historical one by several percentage points suggests a possibility of mispricing.
}
\label{revol}
\end{figure}

The issue of quantifying the gain/loss from trading with a misspecified model has so far mainly been studied in the case of the Black-Scholes model with misspecified volatility \cite{robustbs,romagnoli.vargiolu.00}. In this paper we go one step further, and analyze the effects of misspecification of the volatility itself, the volatility of volatility and of the correlation between the underlying asset and the volatility in a stochastic volatility model. Since these parameters may be observed from a single trajectory of the underlying in an almost sure way, their misspecification leads, in principle, to an arbitrage opportunity. The questions are whether this opportunity can be realized with a feasible strategy, and how to construct a strategy maximizing the arbitrage gain under suitable conditions guaranteeing the well-posedness of the optimization problem. 

While the issue of consistency between real-world and risk-neutral probability measures has been given a rigorous treatment in several papers \cite{awy01,bakshi,dumas}, the corresponding arbitrage trading strategies are usually constructed in an ad-hoc manner \cite{awy01,aqf02,javaheri.05}. For instance, when the risk-neutral skewness is greater than the historical one (which roughly corresponds to correlation misspecification in a stochastic volatility model), \cite{awy01} suggest a strategy consisting in buying all OTM puts and selling all OTM calls. Similarly, if the risk-neutral kurtosis is greater than the historical one, the trading strategy consists in selling far OTM and ATM options while simultaneously buying near OTM options.  In this paper we determine exactly which options must be bought and sold to maximize arbitrage gains, depending on model parameters. 

Our second objective is to analyze commonly used option trading strategies, such as butterflies and risk reversals, and provide a new interpretation of these structures in terms of their performance for volatility arbitrage. A butterfly (BF) is a common strategy in FX option trading, which consists in buying an out of the money call and an out of the money put with the same delta value (in absolute terms) and selling a certain number of at the money calls/puts. A risk reversal (RR) is a strategy consisting in buying an out of the money call and selling an out of the money put with the same delta value (in absolute terms). The financial engineering folklore goes that ``butterflies can be used to arbitrage misspecified volatility of volatility'' and ``risk reversals can be used to arbitrage misspecified correlation''.  In section \ref{bs.sec}, we study these strategies and discuss their optimality for volatility trading.

During the last decade we have witnessed the appearence of a large spectrum of new products specifically designed for volatility and correlation trading, such as variance, volatility and correlation swaps. However, in most markets, European options continue to be much more liquid than exotic volatility products and still constitute the most commonly used tool for volatility arbitrage. In this paper we therefore concentrate on arbitrage strategies involving only the underlying asset and liquid European options. 

The rest of the paper is structured as follows. In Section \ref{setting.sec}, we introduce the generic misspecified stochastic volatility model. Section \ref{portfolio.sec} defines the admissible trading strategies and establishes the form of the optimal arbitrage portfolio. Section \ref{bs.sec} is dedicated to the special case when the misspecified model is the constant volatility Black-Scholes model. This allows us to give a new interpretation of butterfly spreads and risk reversals in terms of their suitability for volatility arbitrage. Section \ref{sabr.sec} presents a simulation study of the performance of the optimal arbitrage strategies in the framework of the SABR stochastic volatility model \cite{sabr}. The purpose of this last section is not to prove the efficiency of our strategies in real markets but simply to provide an illustration using simulated data. A comprehensive empirical study using real market prices is left for further research.

\section{A misspecified stochastic volatility framework}
\label{setting.sec}

We start with a filtered probabiblity space $(\Omega,\mathcal F,\mathbb P, (\mathcal F_t)_{t\geq 0})$ and consider a financial market where there is a risky asset $S$, a risk-free asset and a certain number of European options on $S$. We assume that the interest rate is zero, and that the price of the risky asset $S$ satisfies the stochastic differential equation
\begin{align}
dS_{t}/S_{t}&=\mu_{t}dt+\sigma_t\sqrt{1-\rho^2_t} dW_{t}^{1} + \sigma_t\rho_t dW_{t}^{2} \label{eq:basicsde1}
\end{align}
where $\mu$, $\sigma$ and $\rho \in [-1,1]$ are adapted processes such that 
$$
\int_0^t (1+S_s^2)(1+\mu_s^2 + \sigma_s^2)ds <\infty \quad \text{a.s.}\quad \text{for all $t$,}
$$
and $(W^1,W^2)$ is a standard $2$-dimensional Brownian motion. This integrability condition implies in particular that the stock price process never hits zero $\mathbb P$-a.s.

To account for a possible misspecification of the instantaneous volatility, we introduce the process $\tilde \sigma_t$, which represents the instantaneous volatility  used by the option's market for all pricing purposes. In particular, it is the implied volatility of very short-term at the money options, and in the sequel we call $\tilde \sigma$ the \emph{instantaneous implied volatility} process.  We assume that $\tilde \sigma_t = \tilde \sigma(Y_t)$, where $Y$ is a stochastic process with dynamics
\begin{align}
dY_t&=a_{t}dt+b_{t}dW_{t}^{2},\label{eq:basicsde2}
\end{align}
where $a_t$ and $b_t>0$ are adapted processes such that 
$$
\int_0^t (a_s^2 + b_s^2)ds <\infty \quad \text{a.s.}\quad \text{for all $t$,}
$$ 
and $\tilde\sigma:\mathbb R\to (0,\infty)$ is a continuously differentiable Lipschitz function with $0<\underline\sigma \leq  \tilde\sigma(y)\leq \overline \sigma <\infty$ and $\tilde\sigma'(y)>0$ for all $y\in \mathbb R$;   

Further, to account for possible misspecification of the volatility of volatility $b$ and of the correlation $\rho$, we assume that there exists another probability measure $\mathbb Q$, called market or pricing probability, not necessarily equivalent to $\mathbb P$, such that all options on $S$ are priced in the market as if they were martingales under $\mathbb Q$. The measure $\mathbb Q$ corresponds to the pricing rule used by the market participants, which may be inconsistent with the real-world dynamics of the underlying asset (meaning that $\mathbb Q$ is not necessarily absolutely continuous with respect to $\mathbb P$). Under $\mathbb Q$, the underlying asset and its volatility form a 2-dimensional Markovian diffusion:
\begin{align}
dS_{t}/S_{t}&=\tilde\sigma(Y_t)\sqrt{1-\tilde\rho^2(Y_t,t)} d\tilde W_{t}^{1} + \tilde\sigma(Y_t)\tilde\rho(Y_t,t) d\tilde W_{t}^{2}\label{eq:qsde1}\\
dY_t &=\tilde a(Y_t,t)dt+\tilde b(Y_t,t) d\tilde W_{t}^{2},\label{eq:qsde2}
\end{align}
where the coefficients $\tilde a$, $\tilde b$ and $\tilde \rho$ are deterministic functions and $(\tilde W^1,\tilde W^2)$ is a standard 2-dimensional Brownian motion under $\mathbb Q$. Since $\tilde\sigma$ is bounded, the stock price process never hits zero $\mathbb Q$-a.s.

The following assumptions on the coefficients of \eqref{eq:qsde1}--\eqref{eq:qsde2} will be used throughout the paper: 
\begin{itemize}
\item[i)] There exists $\varepsilon>0$ such that $\min(1-\tilde\rho(y,t)^2,\tilde b(y,t))\geq \varepsilon$ for all $(y,t)\in \mathbb R\times [0,T]$. 
\item[ii)] The functions $\tilde a(y,t)$, $\tilde b(y,t)$, $\tilde \rho(y,t)$ are twice differentiable with respect to $y$; these coefficients as well as their first and second derivatives with respect to $y$ are bounded and Hölder continuous in $y,t$.
\item[iii)] The function $\tilde \sigma$ is twice differentiable; this function as well as its first and second derivative is bounded and Hölder continuous.
\end{itemize}

We suppose that a continuum of European options (indifferently calls or puts) for all strikes and at least one maturity is quoted in the market. The price of an option with maturity date $T$ and pay-off $H(S_T)$ can be expressed as a deterministic function of $S_t$, $Y_t$ and $t$:
$$
P(S_t,Y_t,t) = E^Q[H(S_T)|\mathcal F_t]. 
$$ 
Using standard methods (see e.g. \cite{friedman.75}) one can show that under our assumptions, for every such option, the pricing function $P$ belongs to the class $C^{2,2,1}((0,\infty)\times \mathbb R \times [0,T))$ and satisfies the PDE
\begin{align}
\tilde a \frac{\partial P}{\partial y} + \tilde {\mathcal L} P = 0,\label{pde}
\end{align}
where we define
\[
\tilde {\cal L}f=\frac{\partial f}{\partial t}+\frac{S^{2}\tilde\sigma(y)^{2}}{2}\frac{\partial^{2}f}{\partial S^{2}}+\frac{\tilde b^{2}}{2}\frac{\partial^{2}f}{\partial y^{2}}+S\tilde\sigma(y) \tilde b\tilde \rho\frac{\partial^{2}f}{\partial S\partial y}.\]
In addition (see \cite{romano.touzi}), the price of any such European option satisfies
\begin{align}
\frac{\partial P}{\partial y}>0,\quad \forall (S,y,t) \in (0,\infty)\times \mathbb R \times [0,T).\label{complete}
\end{align}

We shall use the following decay property of the derivatives of call and put prices (see Appendix A for the proof).
\begin{lemma}\label{limits.lm}
Let $P$ be the price of a call or a put option with strike $K$ and maturity date $T$. Then
\begin{align*}
&\lim_{K\to +\infty} \frac{\partial P(S,y,t)}{\partial y} = \lim_{K\to 0} \frac{\partial P(S,y,t)}{\partial y}=0,\\
&\lim_{K\to +\infty} \frac{\partial^2 P(S,y,t)}{\partial y^2} = \lim_{K\to 0} \frac{\partial^2 P(S,y,t)}{\partial y^2} = 0,\\
&\lim_{K\to +\infty} \frac{\partial^2 P(S,y,t)}{\partial S^2} = \lim_{K\to 0} \frac{\partial^2 P(S,y,t)}{\partial S^2} = 0,\\
&\text{and}\quad \lim_{K\to +\infty} \frac{\partial^2 P(S,y,t)}{\partial S\partial y} = \lim_{K\to 0} \frac{\partial^2 P(S,y,t)}{\partial S\partial y} = 0
\end{align*}
for all $(y,t)\in \mathbb R\times [0,T)$. All the above derivatives are continuous in $K$ and the limits are uniform in $S,y,t$ on any compact subset of $(0,\infty)\times \mathbb R \times [0,T)$.  
\end{lemma}

\section{The optimal arbitrage portfolio}
\label{portfolio.sec}
We study the arbitrage from
the perspective of the trader, who knows that the market is using
a misspecified stochastic volatility model to price the options.  We assume full observation: at every date $t$, the trader possesses the information given by the $\sigma$-field $\mathcal F_t$ and knows the deterministic functions $\tilde \sigma$, $\tilde \rho$, $\tilde a$ and $\tilde b$. In Section \ref{sabr.sec} we test the robustness of our results with respect to this assumption. 

To benefit from the market misspecification, our informed trader sets up a dynamic self-financing delta and vega-neutral portfolio
$X_{t}$ with zero initial value, containing, at each date $t$, a stripe of European call or put options with a common expiry date $T$.  In addition, the portfolio contains a quantity $-\delta_t$ of stock and some amount $B_t$ of cash.

To denote the quantity of options of each strike, we introduce a predictable process $(\omega_t)_{t\geq 0}$ taking values in the space $M$ of signed measures on $[0,\infty)$ equipped with the total variation norm $\|\cdot\|_V$. We refer to \cite{bjork97} for technical details of this construction and rigorous definitions of the stochastic integral and the self-financing portfolio in this setting.  Loosely speaking, $\omega_t(dK)$ is the quantity of options with strikes between $K$ and $K+dK$, where a negative measure corresponds to a short position. We shall see later that for the optimal strategies the  measure $\omega$ is concentrated on a finite number of strikes.

The quantity of options of each strike is continuously readjusted meaning that old options are sold to buy new ones with different strikes. In practice, this readjustment will of course happen at discrete dates due to transaction costs and we analyze the performance of our strategies under discrete rebalancing in section \ref{sabr.sec}. The portfolio is held until date $T^*<T$. 

Finally, unless explicitly mentioned otherwise, we assume that at all times, the total quantity of options of all strikes in the portfolio is equal to $1$:
\begin{align}
\|\omega_t\|_V \equiv \int |\omega_t(dK)| \equiv 1\label{int1}
\end{align}
for all $t$. This position constraint ensures that the profit maximization problem is well posed in the presence of arbitrage opportunities (otherwise the profit could be increased indefinitely by increasing the size of positions). If the portfolio is rebalanced completely at discrete equally spaced dates, and the transaction cost per one unit of option is constant, this assumption implies that independently of the composition of the portfolio, the transaction cost paid at each readjustment date is constant, and therefore does not influence the optimization problem. This constraint is also natural from the point of view of an option exchange or a market maker who wants to structure standardized option spreads to satisfy a large number of retail clients, because the number of options in such a spread must be fixed. 

\begin{remark}
For an individual
  trader, who is trying to optimize her option portfolio, it may also be natural to impose a \emph{margin constraint}, that is, a constraint on the capital which the trader uses to meet the margin requirements of the exchange. We discuss option trading under margin constraint in Section \ref{margin.sec}.
\end{remark}

The value of the resulting portfolio
is,
\[
X_{t}=\int P^K(S_t,Y_t,t) \omega_t(dK) -\delta_t S_t + B_t,\]
where the subscript $K$ denotes the strike of the option. 
Together with Lemma \ref{limits.lm}, our assumptions ensure (see \cite[section 4]{bjork97}) that the dynamics of this portfolio are,
\begin{align*}
dX_{t}&=\left(\int \omega_t(dK)  {\cal L}P^K\right) dt+\left(\int \omega_t(dK)\frac{\partial P^K}{\partial S}\right)dS_t\\&+\left(\int \omega_t(dK)\frac{\partial P^K}{\partial y}\right)dY_t-\delta_t dS_t
\end{align*}
where \[
{\cal L}f=\frac{\partial f}{\partial t}+\frac{S_t^{2}\sigma_t^{2}}{2}\frac{\partial^{2}f}{\partial S^{2}}+\frac{b_t^{2}}{2}\frac{\partial^{2}f}{\partial y^{2}}+S_t\sigma_t b_t\rho_t\frac{\partial^{2}f}{\partial S\partial y}\]

To make the portfolio instantaneosly risk-free, we choose, \[
\int \omega_t(dK)\frac{\partial P^K}{\partial y}=0,\qquad \int \omega_t(dK) \frac{\partial{P^K}}{\partial S}=\delta_t\]
to eliminate the $d Y_t$ and $dS_t$ terms. We emphasize that these
hedge ratios are always computed using the market parameters i.e. $\tilde{a}$, $\tilde{b}$ and $\tilde{\rho}$. The portfolio dynamics then become
\begin{eqnarray}
dX_{t}= & \int\omega_t(dK) {\cal L}{P^K}dt,\label{eq:truehedgineq}\end{eqnarray}
and substituting the equation \eqref{pde} into this formula and using the vega-neutrality, we obtain the risk-free profit from model misspecification:
\begin{equation}
dX_t=\int\omega_t(dK) ({\cal L}-\tilde{{\cal L}})P^K dt.\label{eq:payoff0}
\end{equation}
or, at the liquidation date $T^*$,
\begin{equation}
X_{T^*}=\int_{0}^{T^*} \int\omega_t(dK) ({\cal L}-\tilde{{\cal L}})P^K dt,\label{eq:payoff1}\end{equation}
where,
\begin{align}
({\cal L}-\tilde{{\cal L}})P^K&=\frac{S_t^2(\sigma_t^2 - \tilde \sigma^2(Y_t))}{2}\frac{\partial^2 P^K}{\partial S^2} + \frac{(b_t^{2}-\tilde{b}_t^{2})}{2}\frac{\partial^{2}{P^K}}{\partial y^{2}}\notag\\&+S_t(\sigma_t b_t\rho_t-\tilde \sigma(Y_t)\tilde{b}_t\tilde{\rho}_t)\frac{\partial^{2}{P^K}}{\partial S\partial y}\label{diffop}
\end{align}
To maximize the arbitrage profit at each date, we therefore need to solve the following optimization problem:
\begin{align}
&\text{Maximize}\quad \mathcal P_t = \int\omega_t(dK) ({\cal L}-\tilde{{\cal L}})P^K\label{arbprofit}\\
&\text{subject to}\quad \int |\omega_t(dK)|=1 \quad \text{and} \quad \int \omega_t(dK)\frac{\partial P^K}{\partial y}=0.\label{constraints}
\end{align}

\begin{remark}
It should be pointed out that in this paper the arbitrage portfolio is required to be instantaneously risk-free and the performance is measured in terms of instantaneous arbitrage profit. This corresponds to the standard practice of option trading, where the trader is usually only allowed to hold delta and vega netral portfolios and the performance is evaluated over short time scales.  This rules out strategies which are locally risky but have a.s. positive terminal pay-off, such as the strategies based on strict local martingales, which may be admissible under the standard definition of arbitrage. Optimality of arbitrage strategies for an equity market under the standard definition of arbitrage has recently been studied in \cite{fernholz.karatzas.08}. 
\end{remark}

The following result shows that a spread of only two options (with strikes and weights continuously readjusted) is sufficient to solve the problem \eqref{arbprofit}. 

\begin{propo}\label{genopt.prop}
The instantaneous arbitrage profit \eqref{arbprofit}
is maximized by
\begin{align}
\omega_t(dK) = w^1_t \delta_{K^1_t}(dK) - w^2_t \delta_{K^2_t}(dK), \label{optsol}
\end{align}
where $\delta_K(dK)$ denotes the unit point mass at $K$, $(w^1_t,w^2_t)$ are time-dependent optimal weights given by
$$
w^1_t = \frac{\frac{\partial P^{K_2}}{\partial y}}{\frac{\partial P^{K_1}}{\partial y}+\frac{\partial P^{K_2}}{\partial y}},\qquad w^2_t = \frac{\frac{\partial P^{K_1}}{\partial y}}{\frac{\partial P^{K_1}}{\partial y}+\frac{\partial P^{K_2}}{\partial y}},
$$
and $(K^1_t,K^2_t)$ are time-dependent optimal strikes given by
\begin{align}
(K^1_t,K^2_t) = \arg\max_{K^1,K^2} \frac{\frac{\partial P^{K^2}}{\partial y}({\cal L}-\tilde{{\cal L}})P^{K^1} -\frac{\partial P^{K^1}}{\partial y}({\cal L}-\tilde{{\cal L}})P^{K^2}}{\frac{\partial P^{K^1}}{\partial y}+\frac{\partial P^{K^2}}{\partial y}}.\label{2strikes}
\end{align}
\end{propo}
\begin{proof}
\emph{Step 1.} We first show that the optimization problem \eqref{2strikes} is well-posed, that is, the maximum is attained for two distinct strike values. Let $F(K^1,K^2)$ denote the function to be optimized in \eqref{2strikes}. From the property \eqref{complete} and Lemma \ref{limits.lm} it follows that $F$ is continuous in $K_1$ and $K_2$. Let us show that for every $\varepsilon>0$ we can find an interval $[a,b]$ such that $|F(K^1,K^2)|\leq \varepsilon$ for all $(K^1,K^2)\notin [a,b]^2$. We introduce the shorthand notation
\begin{align*}
&g(K) = \frac{\partial P^K}{\partial y},\quad f(K) = ({\cal L}-\tilde{{\cal L}})P^K,\\ 
&F(K^1,K^2) = \frac{g(K^2)f(K^1)-g(K^1)f(K^2)}{g(K^1) + g(K^2)}.
\end{align*}
Fix $\varepsilon>0$. By Lemma \ref{limits.lm}, there exists $N>0$ with $|f(K)|\leq \varepsilon$ for all $K: K\notin [e^{-N},e^{N}]$. Let $\delta = \inf_{|log K|\leq N} g(K)$. Property \eqref{complete} entails that $\delta>0$. Then we can find $M>0$ such that $g(K)\leq \frac{\varepsilon \delta}{sup_{|log K|\leq N}|f(K)|}$ for all $K: K\notin [e^{-M},e^{M}]$ (if the denominator is zero, any $M>0$ can be used). It follows that for all $(K^1,K^2)\notin  [e^{-N}, e^N]\times [e^{-M},e^M]$, $\left|\frac{g(K^2)f(K^1)}{g(K^1) + g(K^2)}\right|\leq \varepsilon$. In the same manner, we can find a rectangle for the second term in $F$, and, taking the square containing both rectangles, we get the desired result. 

Suppose now that for some $K^1$ and $K^2$, $F(K^1,K^2)>0$. Then, by taking $\varepsilon$ sufficiently small, the above argument allows to find a compact set containing the maximum of $F$, and hence the maximum is attained for two strikes which cannot be equal since $F(K,K)\equiv 0$. If, on the other hand, $F(K^1,K^2)\leq 0$ for all $K_1,K_2$, then this means that $F(K^1,K^2)\equiv 0$ and any two strikes can be chosen as maximizers of \eqref{2strikes}. 

\emph{Step 2.} We now show that the two-point solution suggested by this proposition is indeed the optimal one. Let $\omega_t$ be any measure on $\mathbb R^+$ satisfying the constraints \eqref{constraints}. Let $\omega_t = \omega^+_t - \omega^-_t$ be the Jordan decomposition of $\omega_t$ and define $\nu^+_t := \frac{\omega^+_t}{\omega^+_t(\mathbb R^+)}$ and $\nu^-_t := \frac{\omega^-_t}{\omega^-_t(\mathbb R^+)}$. These measures are well-defined due to conditions \eqref{complete} and \eqref{constraints}. From the same conditions,
\begin{align*}
\omega^-_t(\mathbb R^+) = \frac{\int \frac{\partial P^K}{\partial y}d\nu^+_t}{\int \frac{\partial P^K}{\partial y}d\nu^+_t + \int \frac{\partial P^K}{\partial y}d\nu^-_t}, \qquad\omega^+_t(\mathbb R^+) = \frac{\int \frac{\partial P^K}{\partial y}d\nu^-_t}{\int \frac{\partial P^K}{\partial y}d\nu^+_t + \int \frac{\partial P^K}{\partial y}d\nu^-_t}.
\end{align*} 
The instantaneous arbitrage profit \eqref{arbprofit} can then be written as
\begin{align*}
\mathcal P_t &=  \int ({\cal L}-\tilde{{\cal L}})P^Kd\omega^+_t-\int ({\cal L}-\tilde{{\cal L}})P^Kd\omega^-_t\\
&= \frac{\int \frac{\partial P^K}{\partial y}d\nu^-_t \int ({\cal L}-\tilde{{\cal L}})P^K d\nu^+_t-\int \frac{\partial P^K}{\partial y}d\nu^+_t \int ({\cal L}-\tilde{{\cal L}})P^Kd\nu^-_t}{\int \frac{\partial P^K}{\partial y}d\nu^+_t + \int \frac{\partial P^K}{\partial y}d\nu^-_t}.
\end{align*}

Then,
\begin{align*}
\mathcal P_t & = \frac{\int gd\nu^-_t \int f d\nu^+_t-\int gd\nu^+_t \int fd\nu^-_t}{\int gd\nu^+_t + \int gd\nu^-_t},\\
&\leq \frac{\int \nu^-_t(dK) \left\{g(K)+\int g d\nu^+\right\}\sup_{K_1} \frac{g(K_1)\int f d\nu_t^+ - f(K_1)\int g d\nu^+_t}{g(K_1) + \int g d\nu_t^+}}{\int gd\nu^+_t + \int gd\nu^-_t}\\
&= \sup_{K_1} \frac{g(K_1)\int f d\nu_t^+ - f(K_1)\int g d\nu^+_t}{g(K_1) + \int g d\nu_t^+}\\
&\leq \sup_{K_1}\frac{\int \nu^+_t(dK)\left\{g(K_1)+g(K)\right\}\sup_{K_2}\frac{g(K_1)f(K_2)-f(K_1)g(K_2)}{g(K_1)+g(K_2)} }{g(K_1) + \int g d\nu_t^+} \\
&= \sup_{K_1,K_2}\frac{g(K_1)f(K_2)-f(K_1)g(K_2)}{g(K_1)+g(K_2)} 
\end{align*}
Since for the solution \eqref{optsol} the above $\sup$ is attained, this solution is indeed optimal.
\end{proof}

The optimal strategy of Proposition \ref{genopt.prop} is adapted, but not necessarily predictable. In particular, if the maximizer in \eqref{2strikes} is not unique, the optimal strikes may jump wildly between different maximizers leading to strategies which are impossible to implement. However, in practice the maximizer is usually unique (see in particular the next section), and the following result shows that this guarantees the weak continuity, and hence predictability, of the optimal strategy. 

\begin{propo}
Let $t^*<T$, and assume that the coefficients $\sigma_t$, $b_t$ and $\rho_t$ are a.s.~continuous at $t^*$ and that the maximizer in the optimization problem \eqref{2strikes} at $t^*$ is a.s.~unique. Then the optimal strikes $(K^1_t,K^2_t)$ are a.s.~continuous at $t^*$, and the optimal strategy $(\omega_t)$ is a.s.~weakly continuous at $t^*$.  
\end{propo}
\begin{proof}
To emphasize the dependence on $t$, we denote by $F_t(K_1,K_2)$ the function being maximized in \eqref{2strikes}. By the continuity of $\sigma_t$, $b_t$ $\rho_t$, $S_t$ and $Y_t$, there exists a neighborhood $(t_1,t_2)$ of $t^*$ (which may depend on $\omega\in \Omega$), in which these processes are bounded. Since the convergence in Lemma \ref{limits.lm} is uniform on compacts, similarly to the proof of Proposition \ref{genopt.prop}, for every $\varepsilon>0$, we can find an interval $[a,b]$ (which may depend on $\omega \in \Omega$) such that $|F_t(K^1,K^2)|\leq \varepsilon$ $\forall (K^1,K^2) \notin [a,b]^2$ and $\forall t \in (t_1,t_2)$. This proves that the optimal strikes $(K^1_t,K^2_t)$ are a.s.~ bounded in the neighborhood of $t^*$. 

Let $(t_n)_{n\geq 1}$ be a sequence of times converging to $t^*$ and $(K^1_{t_n},K^2_{t_n})_{n\geq 1}$ the corresponding sequence of optimal strikes. By the above argument, this sequence is bounded and a converging subsequence $(K^1_{t_{n_k}},K^2_{t_{n_k}})_{k\geq 1}$ can be extracted. Let $(\tilde K^1, \tilde K^2)$ be the limit of this converging subsequence. Since $(K^1_{t_{n_k}},K^2_{t_{n_k}})$ is a maximizer of $F_{t_{n_k}}$ and by continuity of $F$, for arbitrary strikes $(k^1,k^2)$, 
$$
F_{t^*}(\tilde K^1, \tilde K^2) = \lim_{k\to \infty} F_{t_{n_k}} (K^1_{t_{n_k}},K^2_{t_{n_k}}) \geq \lim_{k\to \infty} F_{t_{n_k}} (k^1,k^2) = F_{t^*}(k^1,k^2).
$$ 
This means that $(\tilde K^1, \tilde K^2)$ is a maximizer of $F_{t^*}(\cdot,\cdot)$, and, by the uniqueness assumption, $\tilde K^1 = K^1_{t^*}$ and $\tilde K^2 = K^2_{t^*}$. By the subsequence criterion we then deduce that $K^1_t \to K^1_{t^*}$ and $K^2_t \to K^2_{t^*}$ a.s. as $t\to t^*$. The weak convergence of $\omega_t$ follows directly. 
\end{proof}

\subsection{Trading under margin constraint}
\label{margin.sec}
Solving the optimization problem \eqref{arbprofit} under the constraint \eqref{int1} amounts to finding the optimal (vega-weighted) spread, that is a vega-neutral combination of two options, in which the total quantity of options equals one. This approach guarantees the well-posedness of the optimization problem and provides a new interpretation for commonly traded option contracts such as butterflies and risk reversals (see next section). 

From the point of view of an individual option trader, an interesting alternative to \eqref{int1} is the constraint given by the margin requirements of the stock exchange. In this section, we shall consider in detail the CBOE minimum margins for customer accounts, that is, the margin requirements for retail investors imposed by the Chicago Board of Options Exchange \cite{cboe.margin,santa-clara.saretto.09}. Under this set of rules, margin requirements are, by default, calculated individually for each option position: 
\begin{itemize}
\item Long positions in call or put options must be paid for in full.
\item For a naked short position in a call option, the margin requirement at time $t$ is calculated using the following formula \cite{santa-clara.saretto.09}.
\begin{align*}
M_t& = P^K_t + \max(\alpha S_t - (K-S_t)1_{K>S_t}, \beta S_t) := P^K_t + \lambda^K_t, 
\end{align*}
where $P^K$ is the option price, $\alpha=0.15$ and $\beta = 0.1$.  
\item The margin requirement for a ``long call plus short underlying position'' is equal to the call price plus  short sale proceeds plus $50\%$ of the underlying value. 
\item The margin requirement for a "short call plus long underlying" is $50\%$ of the underlying value.
\end{itemize}
In addition, some risk offsets are allowed for certain spread options, but we shall not take them into account to simplify the treatment. For the same reason, we shall assume that all options held by the trader are call options. Let $\omega_t = \omega^+_t - \omega^-_t$ be the Jordan decomposition of $\omega_t$. The value of a delta-hedged option portfolio can be written as
\begin{align*}
X_{t}=B_t &+ \int \left\{P^K -S_t \frac{\partial P^K}{\partial S} \right\}\omega_t(dK)\\
 = B_t &+ \int \left\{P^K\left(1 - \frac{\partial P^K}{\partial S} \right) + \frac{\partial P^K}{\partial S} (P^K-S_t)\right\}\omega^+_t(dK) \\
&+ \int \left\{-P^K\left(1 - \frac{\partial P^K}{\partial S} \right) + \frac{\partial P^K}{\partial S} (S_t-P^K)\right\}\omega^-_t(dK)
\end{align*}
The margin requirement for this position is 
\begin{align*}
M_t &= \int \left\{P^K + \frac{\partial P^K}{\partial S} (1+\gamma)S_t \right\}\omega^+_t(dK)\\&+\int \left\{(\lambda^K + P^K)\left(1 - \frac{\partial P^K}{\partial S} \right) + \gamma S_t \frac{\partial P^K}{\partial S}\right\}\omega^-_t(dK)
\\ & := \int \beta^+_t(K) \omega^+_t(dK) + \int \beta^-_t(K) \omega^-_t(dK)
\end{align*}
with $\gamma=0.5$. 
Supposing that the trader disposes of a fixed margin account (normalized to $1$), the optimization problem \eqref{arbprofit}--\eqref{constraints} becomes
\begin{align*}
&\text{Maximize}\quad \mathcal P_t = \int\omega_t(dK) ({\cal L}-\tilde{{\cal L}})P^K\\
&\text{subject to}\quad \int \beta^+_t(K) \omega^+_t(dK) + \int \beta^-_t(K) \omega^-_t(dK) = 1 \\ &\text{and} \quad \int \omega_t(dK)\frac{\partial P^K}{\partial y}=0.
\end{align*}

This maximization problem can be treated similarly to \eqref{arbprofit}--\eqref{constraints}. Assume that at time $t$, the problem 
\begin{align}
\arg\sup_{K^1,K^2} \frac{\frac{\partial P^{K^2}}{\partial y}({\cal L}-\tilde{{\cal L}})P^{K^1} -\frac{\partial P^{K^1}}{\partial y}({\cal L}-\tilde{{\cal L}})P^{K^2}}{\beta^-(K^2)\frac{\partial P^{K^1}}{\partial y}+\beta^+(K^1)\frac{\partial P^{K^2}}{\partial y}}.\label{2strikesbeta}
\end{align}
admits a finite maximizer $(K^1_t,K^2_t)$. Then the instantaneous arbitrage profit is maximized by the two-strike portfolio \eqref{optsol} with
$$
w^1_t = \frac{\beta^+(K^1_t)\frac{\partial P^{K_2}}{\partial y}}{\beta^+(K^2_t)\frac{\partial P^{K_1}}{\partial y}+\beta^+(K^1_t)\frac{\partial P^{K_2}}{\partial y}},\qquad w^2_t = \frac{\beta^+(K^2_t)\frac{\partial P^{K_1}}{\partial y}}{\beta^+(K^2_t)\frac{\partial P^{K_1}}{\partial y}+\beta^+(K^1_t)\frac{\partial P^{K_2}}{\partial y}}.
$$
However, since $\beta^+(K)\to 0$ as $K\to +\infty$, the problem \eqref{2strikesbeta} may not admit a finite maximizer for some models. For instance, in the Black-Scholes model,
a simple asymptotic analysis shows that $\beta^+(K)$ converges to zero faster than $\frac{\partial P^K}{\partial y}$ as $K\to +\infty$, while  $\frac{\partial P^K}{\partial y}$ converges to zero faster than $\frac{\partial^2 P^K}{\partial S\partial y}$ and $\frac{\partial^2 P^K}{\partial y^2}$. This means that by keeping $K^2$ fixed and choosing $K^1$ large enough, the instantaneous arbitrage profit \eqref{2strikesbeta} can be made arbitrarily large. In financial terms this means that since the margin constraint does not limit the size of the long positions, it may be optimal to buy a large number of far out of the money options. So the margin constraint alone is not sufficient to guarantee the well-posedness of the problem and a position constraint similar to \eqref{int1} is necessary as well (in particular, the position constraint limits the transaction costs). If both the margin constraint and the position constraint are imposed, the optimal portfolio will in general be a combination of three options
with different strikes. 

While our analysis focused on CBOE margining rules for retail investors, a number of other exchanges, for example, CME for large institutional accounts, use an alternative margining system called Standard Portfolio Analysis of Risk (SPAN). Under this system, the margin requirement is determined as the worst-case loss of the portfolio over a set of 14 market scenarios. In this case, position constraints are also necessary to ensure the well-posedness of the optimization problem. See \cite{avriel.reisman.00} for a numerical study of optimal option portfolios under SPAN constraints in the binomial model.
\section{The Black-Scholes case}
\label{bs.sec}
In this section, we consider the case when the misspecified model is the Black-Scholes model, that is, $\tilde a \equiv \tilde b  \equiv 0$. This means that at each date $t$ the market participants price options using the Black-Scholes model with the volatility value given by the current instantaneous implied volatility $\tilde \sigma(Y_t)$. The process $Y_t$ still has a stochastic dynamics under the real-world probability $\mathbb P$. This is different from the setting of \cite{robustbs}, where one also has $a=b=0$, that is, the instantaneous implied volatility is deterministic. 
In this section, since we do not need to solve partial differential equations, to simplify notation, we set $\tilde \sigma_t \equiv Y_t$, that is, the mapping $\tilde \sigma(\cdot)$ is the identity mapping.
 Recall the formulas for the derivatives of the call/put option price in the Black-Scholes model ($r=0$):
\begin{align*}
&\frac{\partial P}{\partial \sigma} = Sn(d_1)\sqrt{\tau} = Kn(d_2)\sqrt{\tau},\qquad &&\frac{\partial^2 P}{\partial S^2} = \frac{n(d_1)}{S\tilde\sigma\sqrt{\tau}},\\
&\frac{\partial^2 P}{\partial \sigma \partial S} = -\frac{n(d_1) d_2}{\tilde\sigma},\qquad 
&&\frac{\partial^2 P}{\partial \sigma^2} = \frac{S n(d_1)d_1 d_2\sqrt{\tau}}{\tilde\sigma},
\end{align*}
where $d_{1,2} = \frac{m}{\tilde\sigma\sqrt{\tau}}\pm \frac{\tilde\sigma\sqrt{\tau}}{2}$, $\tau= T-t$, $m = \log(S/K)$ and $n$ is the standard normal density.

We first specialize the general optimal solution \eqref{optsol} to the Black-Scholes case. All the quantities $\sigma_t,\tilde\sigma_t,b_t,\rho_t,S_t$ are, of course, time-dependent, but since the optimization is done at a given date, to simplify notation, we omit the subscript $t$. 
\begin{propo}\label{bsopt.prop}
The optimal option portfolio maximizing the instantaneous arbitrage profit \eqref{arbprofit} is described as follows:
\begin{itemize}
\item The portfolio consists of a long position in an option with log-moneyness $m_1 = z_1 \tilde\sigma\sqrt{\tau} - \frac{\tilde\sigma^2\tau}{2}$ and a short position in an option with log-moneyness $m_2 = z_2 \tilde\sigma\sqrt{\tau} - \frac{\tilde\sigma^2\tau}{2}$, where $(z_1,z_2)$ is a maximizer of the function
$$
f(z_1,z_2) = \frac{(z_1-z_2)(z_1+z_2-w)}{e^{z_1^2/2}+e^{z_2^2/2}}
$$
with $w = \frac{\tilde\sigma b\tau+2\sigma\rho}{b\sqrt{\tau}}$.
\item The weights of the two options are chosen to make the portfolio vega-neutral.
\end{itemize}
We define by $P_{opt}$ the instantaneous arbitrage profit realized by the optimal portfolio.
\end{propo}
\begin{proof}
Substituting the Black-Scholes values for the derivatives of option prices, and making the change of variable $z = \frac{m}{\tilde\sigma\sqrt{\tau}}+\frac{\tilde\sigma\sqrt{\tau}}{2}$, the maximization problem to solve becomes:
\begin{align}
\max_{z_1,z_2} bS \frac{n(z_1)n(z_2)}{n(z_1)+n(z_2)}\left\{\frac{b\sqrt{\tau}}{2\tilde\sigma}(z_1^2 - z_2^2) - \frac{b\tau}{2}(z_1-z_2)-\rho\frac{\sigma}{\tilde\sigma}(z_1-z_2)\right\},\label{bsoptim}
\end{align}
from which the proposition follows directly.
\end{proof}
\begin{remark}
A numerical study of the function $f$ in the above proposition shows that it admits a unique maximizer in all cases except when $w=0$. 
\end{remark}
\begin{remark}
From the form of the functional \eqref{bsoptim} it is clear that one should choose options with the largest time to expiry $\tau$ for which liquid options are available. 
\end{remark}
\begin{remark}
The gamma term (second derivative with respect to the stock price) does not appear in the expression \eqref{bsoptim}, because the optimization is done under the constraint of vega-neutrality, and in the Black-Scholes model a portfolio of options with the same maturity is vega-neutral if and only if it is gamma neutral (see the formulas for the greeks earlier in this section). Therefore in this setting, to arbitrage the misspecification of the volatility itself, as opposed to skew or convexity, one needs a portfolio of options with more than one maturity.  
\end{remark}

The variable $z=\frac{\log(S/K)}{\sigma\sqrt{\tau}}+\frac{\sigma\sqrt{\tau}}{2}$ introduced in the proof is directly related to the Delta of the option: for a call option $\Delta = N(z)$. This convenient parameterization corresponds to the current practice in Foreign Exchange markets of expressing the strike of an option via its delta. Given a weighting measure $\omega_t(dK)$ as introduced in section \ref{setting.sec}, we denote by $\bar\omega_t(dz)$ the corresponding measure in the $z$-space. We say that a portfolio of options is $\Delta$-symmetric (resp., $\Delta$-antisymmetric) if $\bar \omega_t$ is symmetric (resp., antisymmetric). 
 
The following result clarifies the role of butterflies and risk reversals for volatility arbitrage. 
\begin{propo}
Let $\mathcal P^{opt}$ be the instantaneous arbitrage profit \eqref{arbprofit} realized by the optimal strategy of Proposition \ref{bsopt.prop}
\begin{enumerate}
\item Consider a portfolio (RR) described as follows:
\begin{itemize}
\item If $b\tau/2+\rho \sigma/\tilde\sigma\geq 0$
\begin{itemize}
\item buy $\frac{1}{2}$ units of options with log-moneyness $m_1 = -\tilde\sigma \sqrt{\tau} -\frac{\tilde\sigma^2 \tau}{2}$, or, equivalently, delta value $N(-1)\approx 0.16$ 
\item sell $\frac{1}{2}$ units of options with log-moneyness  $m_2 = \tilde\sigma \sqrt{\tau} -\frac{\tilde\sigma^2 \tau}{2}$, or, equivalently, delta value $N(1)\approx 0.84$. 
\end{itemize}
\item if $b\tau/2+\rho\sigma/\tilde\sigma < 0$ buy the portfolio with weights of the opposite sign.
\end{itemize} 
Then the portfolio (RR) is the solution of the maximization problem \eqref{arbprofit} under the additional constraint that it is $\Delta$-antisymmetric.   
\item Consider a portfolio (BF) described as follows 
\begin{itemize}
\item buy $x$ units of options with log-moneyness $m_1 = z_0 \tilde\sigma \sqrt{\tau} - \tilde\sigma^2 \tau$, or, equivalently, delta value $N(z_0)\approx 0.055$, where $z_0\approx 1.6$ is a universal constant, solution to 
$$
\frac{z_0^2}{2} e^{\frac{z_0^2}{2}} = e^{\frac{z_0^2}{2}} +1
$$
\item buy $x_0$ units of options with log-moneyness $m_2 = -z_0 \tilde\sigma \sqrt{\tau} - \tilde\sigma^2 \tau$, or, equivalently, delta value $N(-z_0)\approx 0.945$
\item sell $1-2x_0$ units of options with log-moneyness $m_3 = -\frac{\tilde\sigma^2 \tau}{2}$ or, equivalently, delta value $N(0) = \frac{1}{2}$, where 
the quantity $x_0$ is chosen to make the portfolio vega-neutral, that is, $x_0 = \frac{1}{2(1+e^{-z_0^2/2})}\approx 0.39$. 
\end{itemize}
Then, the portfolio (BF) is the solution of the maximization problem \eqref{arbprofit} under the additional constraint that it is $\Delta$-symmetric.
\item Define by $\mathcal P^{RR}$ the instantaneous arbitrage profit realized by the portfolio of part 1 and by $\mathcal P^{BF}$ that of part 2. Let 
$$
\alpha = \frac{|\tilde \sigma b\tau+2\sigma\rho|}{|\tilde \sigma b\tau+2\sigma\rho| + 2b K_0 \sqrt{\tau}}
$$
where $K_0$ is a universal constant, defined below in the proof, and approximately equal to $0.459$. Then
$$
\mathcal P^{RR}\geq \alpha \mathcal P^{opt} \quad \text{and} \quad  \mathcal P^{BF}\geq (1-\alpha) \mathcal P^{opt}.
$$
\end{enumerate}
\end{propo}

\begin{remark}
In case of a risk reversal, the long/short decision depends on the sign of $\rho + \frac{\tilde \sigma b \tau}{2\sigma}$ rather than on the sign of the correlation itself. This means that when the implied volatility surface is flat and $\rho<0$, if $b$ is big enough, contrary to conventional wisdom, the strategy of selling downside strikes and buying upside strikes will yield a positive profit, whereas the opposite strategy will incur a loss. This is because the risk reversal (RR) has a positive vomma (second derivative with respect to Black-Scholes volatility). Let $P^x$ be the price of the risk reversal contract where one buys a call with delta value $x<\frac{1}{2}$ and sells a call with delta value $1-x$. Then
$$
\frac{\partial^2 P^x}{\partial \sigma^2} = S n(d) \tilde \sigma \tau d,\quad d = N(1-x)>0.
$$
Therefore, from equation \eqref{diffop} it is clear that if $b$ is big enough, the effect of the misspecification of $b$ will outweigh that of the correlation leading to a positive instantaneous profit. 
\end{remark}

\begin{proof}
Substituting the Black-Scholes values for the derivatives of option prices, and making the usual change of variable, the maximization problem \eqref{arbprofit} becomes:
\begin{align}
&\max \frac{Sb^2\sqrt{\tau}}{2\tilde\sigma}\int z^2 n(z) \bar\omega_t(dz) - Sb(b\tau/2+\rho \sigma/\tilde\sigma)\int zn(z)\bar\omega_t(dz)\label{bsfunc}\\
&\text{subject to}\quad \int n(z)\bar\omega_t(dz)=0,\quad \int |\bar\omega_t(dz)| = 1. \notag
\end{align}
\begin{enumerate}
\item In the antisymmetric case it is sufficient to look for a measure $\bar \omega^+_t$ on $(0,\infty)$ solution to
\begin{align*}
&\max  - Sb(b\tau/2+\rho \sigma/\tilde\sigma)\int_{(0,\infty)} zn(z)\bar\omega^+_t(dz)\\
&\text{subject to}\quad \int_{(0,\infty)} |\bar\omega_t(dz)| = \frac{1}{2}. \notag
\end{align*}
The solution to this problem is given by a point mass at the point
$z = \arg\max zn(z) = 1$ with weight $-\frac{1}{2}$ if $b\tau/2+\rho \sigma/\tilde\sigma\geq 0$ and  $\frac{1}{2}$ if $b\tau/2+\rho \sigma/\tilde\sigma< 0$. Adding the antisymmetric component of $\bar\omega_t^+$, the proof of part 1 is complete.

\item 
In the symmetric case it is sufficient to look for a measure $\bar \omega^+_t$ on $[0,\infty)$ solution to
\begin{align*}
&\max \frac{Sb^2\sqrt{\tau}}{2\tilde\sigma}\int_{[0,\infty)} z^2 n(z) \bar\omega^+_t(dz) \\
&\text{subject to}\quad \int_{[0,\infty)} n(z)\bar\omega^+_t(dz)=0,\quad \int_{[0,\infty)} |\bar\omega^+_t(dz)| = \frac{1}{2}. \notag
\end{align*}
By the same argument as in the proof of Proposition \ref{genopt.prop}, we get that the solution to this problem is given by two point masses at the points $z_1$ and $z_2$ given by
$$
(z_1,z_2) = \arg\max \frac{n(z_1)z_2^2 n(z_2)-n(z_2)z_1^2 n(z_1)}{n(z_1)+n(z_2)} = \arg\max \frac{z_2^2 -z_1^2 }{e^{\frac{z_1^2}{2}}+e^{\frac{z_2^2}{2}}},
$$
from which we immediately see that $z_1=0$ and $z_2$ coincides with the constant $z_0$ introduced in the statement of the proposition. Adding the symmetric part of the measure $\bar \omega_t^+$, the proof of part 2 is complete. 

\item To show the last part, observe that the contract (BF) maximizes the first term in \eqref{bsfunc} while the contract (RR) maximizes the second term. The values of \eqref{bsfunc} for the contract (BF) and (RR) are given by
$$
\mathcal P^{BF} = \frac{Sb^2 \sqrt{\tau}}{\tilde\sigma \sqrt{2\pi}}e^{-z_0^2/2},\qquad \mathcal P^{RR} = \frac{Sb|b\tau/2+\rho \sigma/\tilde\sigma|}{\sqrt{2\pi}}e^{-\frac{1}{2}}.
$$
and therefore
$$
\frac{\mathcal P^{RR}}{\mathcal P^{BF}+\mathcal P^{RR}} = \frac{|\tilde \sigma b\tau+2\sigma\rho|}{|\tilde \sigma b\tau+2\sigma\rho| + 2b K_0 \sqrt{\tau}} \quad \text{with}\quad K_0 = e^{\frac{1}{2}-\frac{z_0^2}{2}}.
$$
Since the maximum of a sum is always no greater than the sum of maxima, $\mathcal P^{opt} \leq \mathcal P^{BF}+\mathcal P^{RR}$ and the proof is complete. 
\end{enumerate}
\end{proof}

\begin{remark}
Let us sum up our findings about the role of butterflies and risk reversals for arbitrage:
\begin{itemize}
\item Risk reversals are not optimal and butterflies are not optimal unless $\rho = -\frac{b\tilde \sigma\tau}{2\sigma}$ (because $\alpha=0$ for this value of $\rho$).
\item Nevertheless there exists a universal risk reversal (16-delta risk reversal in the language of foreign exchange markets, denoted by RR) and a universal butterfly (5.5-delta vega weighted buttefly, denoted by BF) such that one can secure at least half of the optimal profit by investing into one of these two contracts. Moreover, for each of these contracts, a precise estimate of the deviation from optimality is available. 
\item Contrary to the optimal portfolio which depends on model parameters via the constant $w$ introduced in Proposition \ref{bsopt.prop}, the contracts (BF) and (RR) are model independent and only the choice of the contract (BF or RR) and the sign of the risk reversal may depend on model parameters. This means that these contracts are to some extent robust to the arbitrageur using a wrong parameter, which is important since in practice the volatility of volatility and the correlation are difficult to estimate. 

\item The instantaneous profit $\mathcal P^{BF}$ of the contract (BF) does not depend on the true value of the correlation $\rho$. This means that (BF) is a \emph{pure convexity trade}, and it is clear from the above that it has the highest instantaneous  profit among all such trades. Therefore, the contract (BF) should be used if the sign of the correlation is unknown. 
\item When $b\rightarrow 0, \alpha\rightarrow 1$, and in this case (RR) is the optimal strategy.
\end{itemize}
\end{remark}

\section{Simulation study in the SABR model}
\label{sabr.sec}
In this section, we consider the case when, under the pricing probability, the underlying asset price follows the stochastic volatility model known as SABR model. This model captures both the volatility of volatility and the correlation effects and is analytically tractable. The dynamics of the underlying asset under $\mathbb Q$ is
\begin{align}
dS_t&=\tilde \sigma_t S_t^{\beta}(\sqrt{1-\tilde \rho^2}d\tilde W^1_t + \tilde \rho d\tilde W^2_t)\label{eq:sabrsde1}\\
d\tilde\sigma_t&=\tilde b \tilde \sigma_t d\tilde W^2_t\label{eq:sabrsde2}
\end{align}
To further simplify the treatment, we take $\beta=1$, and in order to guarantee that $S$ is a martingale under the pricing probability \cite{sin.98}, we assume that the correlation coefficient satisfies $\tilde \rho\leq 0$. The true dynamics of the instantaneous implied volatility are 
\begin{align}
d\tilde\sigma_t=b \tilde \sigma_t d W^2_t,\label{eq:sabrsde3}
\end{align}
and the dynamics of the underlying under the real-world measure are
\begin{align}
dS_t= \sigma_t S_t (\sqrt{1- \rho^2}dW^1_t + \rho dW^2_t).\label{eq:sabrsde4}
\end{align}

The SABR model does not satisfy some of the assumptions of section \ref{setting.sec}, in particular because the volatility is not bounded from below by a positive constant. Nevertheless, it provides a simple and tractable framework to illustrate the performance of our strategies. A nice feature of the SABR model is that rather precise approximate pricing formulas for vanilla options are available, see \cite{sabr} and \cite{osajima.07}, which can be used to compute the optimal strikes of Proposition \ref{genopt.prop}. Alternatively, one can directly compute the first order correction (with respect to the volatility of volatility parameter $\tilde b$) to the Black-Scholes optimal values using perturbation analysis. We now briefly outline the main ideas behind this approach.

\paragraph{First order correction to option price} 
In the SABR model, the call/put option price $P$ satisifies the following pricing equation,
\[
\frac{\partial P}{\partial t}+\frac{S^{2}\tilde\sigma^{2}}{2}\frac{\partial^{2}P}{\partial S^{2}}+\frac{\tilde b^2\tilde \sigma^{2}}{2}\frac{\partial^{2}P}{\partial\sigma^{2}}+S\tilde\sigma^2 \tilde b\tilde\rho\frac{\partial^{2}P}{\partial S\partial\sigma}=0\]
with the appropriate terminal condition. We assume that $\tilde b$ is small and look for approximate solutions of the form $P=P_{0}+\tilde b P_{1}+O(\tilde b^2)$. The zero-order term $P_{0}$ corresponds to the Black Scholes solution:
\[
\frac{\partial P_{0}}{\partial t}+\frac{S^{2}\tilde\sigma^{2}}{2}\frac{\partial^{2}P_{0}}{\partial S^{2}}=0\]
The first order term satisfies the following equation:
\[
\frac{\partial P_{1}}{\partial t}+\frac{S^{2}\tilde\sigma^{2}}{2}\frac{\partial^{2}P_{1}}{\partial S^{2}}+S\tilde \sigma^{2}\tilde \rho\frac{\partial^{2}P_{0}}{\partial S\partial\sigma}=0\]
and can be computed explicitly via
$$
P_{1}=\frac{\tilde\sigma^2 \tilde\rho(T-t)}{2}S\frac{\partial^{2}P_{0}}{\partial S\partial\sigma}
$$

\paragraph{First order correction to optimal strikes} Define
\[
f=({\cal L}-\tilde{{\cal L}})P\]
where we recall that
\[
({\cal L}-\tilde{{\cal L}})P=\frac{S^2 (\sigma^2-\tilde\sigma^2)}{2}\frac{\partial^2 P}{\partial S^2} + \frac{(b^{2}-\tilde{b}^{2})\tilde\sigma^2}{2}\frac{\partial^{2}P}{\partial\sigma^{2}}+S(\sigma\tilde \sigma b\rho-\tilde\sigma^2\tilde{b}\tilde{\rho})\frac{\partial P^{2}}{\partial S\partial\sigma}\]
For the general case, we are not assuming anything on the smallness of the parameters of the true model. 
Now to the first order in $\tilde b$, we can expand $f$ as,
\[
f\approx f_{0}+\tilde b f_{1}\]
where,
\[
f_{0}=\frac{S^2 (\sigma^2-\tilde\sigma^2)}{2}\frac{\partial^2 P_0}{\partial S^2} + \frac{b^{2}\tilde\sigma^2}{2}\frac{\partial^{2}P_{0}}{\partial\sigma^{2}}+S\sigma\tilde\sigma b\rho\frac{\partial P_{0}^{2}}{\partial S\partial\sigma}\]
and,
\[
f_{1}=\frac{S^2 (\sigma^2-\tilde\sigma^2)}{2}\frac{\partial^2 P_1}{\partial S^2} +\frac{b^{2}\tilde\sigma^2}{2}\frac{\partial^{2}P_{1}}{\partial\sigma^{2}}+S\sigma\tilde\sigma b\rho\frac{\partial^{2}P_{1}^{2}}{\partial S\partial\sigma}-S\tilde\sigma^2 \tilde{\rho} \frac{\partial^{2}P_{0}^{2}}{\partial S\partial\sigma}\]

Define the expression to be maximized from Proposition \ref{genopt.prop}:
\[
F=\frac{w^{2}f^{1}-w^{1}f^{2}}{w^{1}+w^{2}}\]
where $w^{i}=\frac{\partial P^{i}}{\partial\sigma}$. These vegas are also expanded in the same manner: $w^i \approx w^i_0 + \tilde b w^i_1$. 
Now to first order in $\tilde b$, we can expand $F$ as,
\begin{align}
F\approx \tilde F :=F_{0}+\tilde b(F_{1}+F_{3}F_{0})\label{ftilde}
\end{align}
where
\begin{align*}
F_{0}&=\frac{w_{0}^{2}f_{0}^{1}-w_{0}^{1}f_{0}^{2}}{w_{0}^{1}+w_{0}^{2}}\\
F_{1}&=\frac{w_{0}^{2}f_{1}^{1}+w_{1}^{2}f_{0}^{1}-w_{0}^{1}f_{1}^{2}-w_{1}^{1}f_{0}^{2}}{w_{0}^{1}+w_{0}^{2}}\\
\text{and}\quad F_{3}&=\frac{w_{1}^{1}+w_{1}^{2}}{w_{0}^{1}+w_{0}^{2}}.
\end{align*}

We would like to represent the optimal strike $K^*$ which maximizes $\tilde F$ as the sum of the strike $K_0$ which maximizes $F_0$  plus a small correction of order of $\tilde b$:
$$
K^*\approx K_0 + \tilde b K_1.
$$  
The optimal strike $K^*$ is the solution of the equation,
\[
\frac{\partial \tilde F}{\partial K}(K^*)=0 \quad \iff \quad \frac{\partial \tilde F}{\partial K}(K_0)+\tilde b K_1 \frac{\partial^{2}\tilde F}{\partial K^{2}}(K_0) \approx 0.\]
Using the representation \eqref{ftilde}, the fact that $K_0$ is the maximizer of $F_0$ and keeping only the terms of order up to one in $\tilde b$, we get
\[
\frac{\partial F_{1}}{\partial K}(K_0)+F_{0}\frac{\partial F_{3}}{\partial K}(K_0)+K_1\frac{\partial^{2}F_{0}}{\partial K^{2}}(K_0)=0\]
which gives the following correction to the optimal strike:
\[
K_1=-\frac{\frac{\partial F_{1}}{\partial K}(K_0)+F_{0}\frac{\partial F_{3}}{\partial K}(K_0)}{\frac{\partial^{2}F_{0}}{\partial K^{2}}(K_0)}\]

Figure \ref{pertfig} compares this first-order correction to the two optimal strikes to the zero-order Black-Scholes value and the precise value computed using Hagan's implied volatility asymptotics \cite{sabr}, for different values of the volatility of volatility parameter $\tilde b$. Given that already the difference between the zero-order correction and the precise value is not that big, our first order correction appears acceptable for reasonably small values of $\tilde b$. 

\begin{figure} 
 \centering
 \begin{tabular}{cc}
  \epsfig{file=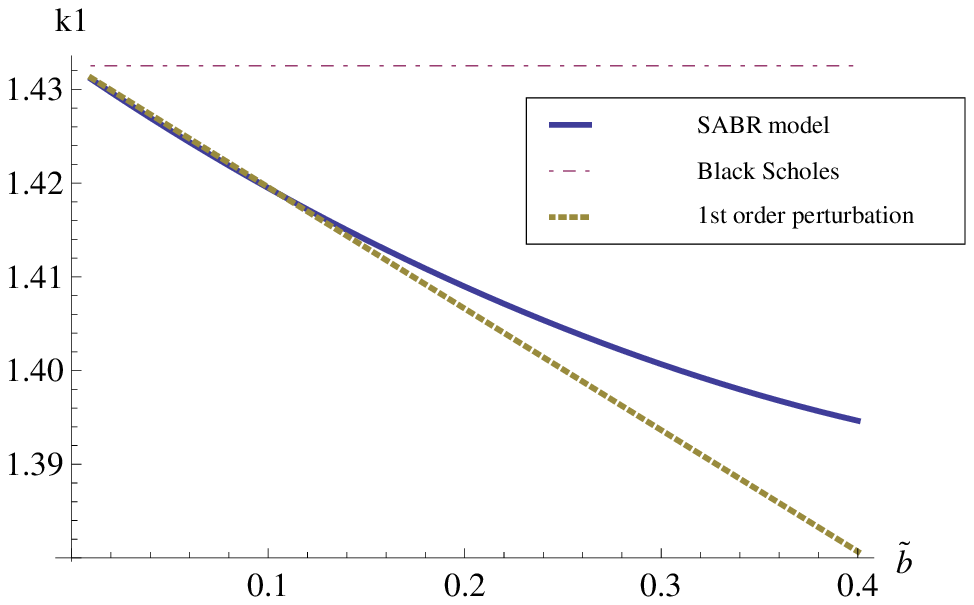,scale=.7} &
\hspace{-20pt}
  \epsfig{file=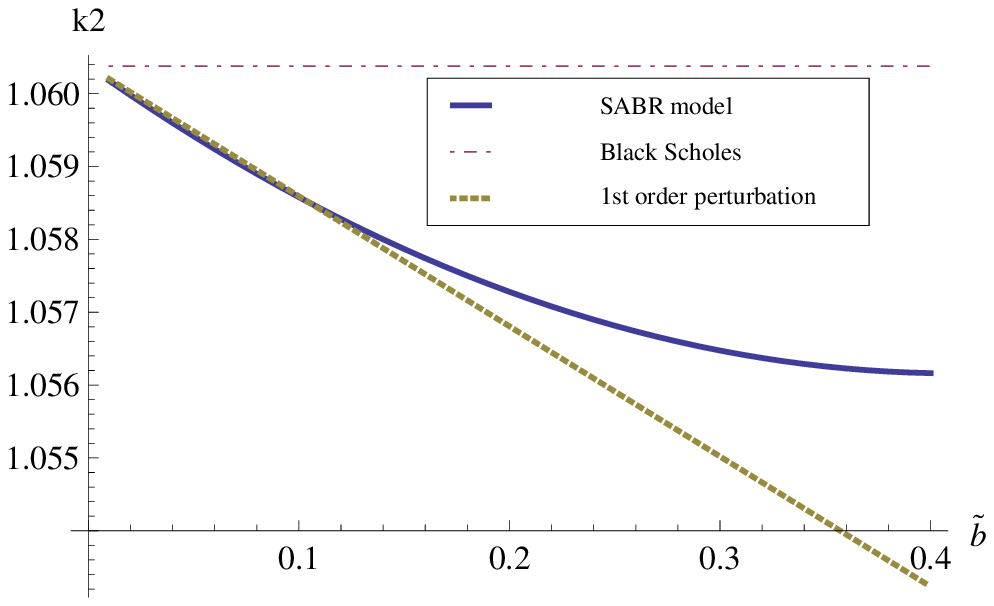,scale=.7} 
 \end{tabular}
\caption{Optimal Strikes for the set of parameters $\sigma=.2$, $S=1, b=.3, \rho=-.3, \tilde{\rho}=-.5, t=1$, as a function of the misspecified $\tilde{b} \in [.01,.4]$. }
\label{pertfig}
\end{figure}

\paragraph{Numerical results}

In this section we demonstrate the performance of our strategies on a simulated example. We assume that a trader, who is aware both about the true and the misspecified parameters, invests in our strategy. We simulate 10000 runs\footnote{For the figure which includes transaction costs, 100000 simulation runs are done.} of the stock and the volatility during 1 year under the system specified by \eqref{eq:sabrsde1}--\eqref{eq:sabrsde4}, assuming misspecification of the correlation $\rho$ and the volatility of volatility $b$ but not of the volatility itself $(\sigma=\tilde\sigma)$.  The initial stock value is $S_0=100$ and initial volatility is $\sigma_0=.1$. The interest rate is assumed to be zero. During the simulation, at each rebalancing date, the option portfolio is completely liquidated, and a new portfolio of options with the same, fixed, time to maturity and desired optimal strikes is purchased. 

We illustrate the impact of different characteristics of the strategy on the overall performance. In each of the  graphs, we plot two groups of curves: the cross-marked ones correspond to misspecified model, and the diamond-market ones correspond to the setting when the model parameters are equal to the true ones. In each group, the three curves correspond (from lower to upper) to the 25\% quantile, the median and the 75\% quantile computed from the 10000/100000 simulated trajectories. 
\begin{itemize}
\item Impact of correlation misspecification. Figure \ref{corrvol.fig} (left) shows the performance of portfolios using options with 1 month to maturity. The true parameters are $\rho=-.8, b=.2$. The misspecified or the market parameters are $\tilde{\rho}=-.2, \tilde{b}=.2$.
\item Impact of volatility of volatility misspecification. Figure \ref{corrvol.fig} (right) shows the performance of portfolios using options with 1 month to maturity. The true parameters are $\rho=-.8, b=.2$. The misspecified or the market parameters are $\tilde{\rho}=-.8, \tilde{b}=.1$.
\item Impact of the options' maturity. Figure \ref{maturity.fig} shows the difference in performance of portfolios using options with 1 month and 6 months to maturity. The true parameters are $\rho=-.8, b=.2$. The misspecified or the market parameters are $\tilde{\rho}=-.2, \tilde{b}=.1$. As predicted by our analysis in the Black-Scholes case, the performance is better for options with longer time to maturity and higher intrinsic value.  
\item Impact of the rebalancing frequency. Figure \ref{frequency.fig} illustrates the effect of the total number of rebalancings on portfolio performance. The true parameters are $\rho=-.8, b=.2$. The misspecified or the market parameters are $\tilde{\rho}=-.2, \tilde{b}=.1$. The average profit is roughly the same for both the rebalancing frequencies but the variance is higher for lower rebalancing frequencies.

\item Impact of trader using different parameters from the true ones. Figure \ref{trader.fig} shows the performance of portfolios using options with 1 month to maturity. The misspecified or the market parameters are $\tilde{\rho}=-.2, \tilde{b}=.1$. The true historical parameters are $\rho=-.8, b=.2$ but the trader estimates these parameters with an error and uses the values $\hat{\rho}=-.6, \hat{b}=.15$ instead. The average profit increases with better estimation of misspecification in parameters.
\item Impact of transaction costs. In all our analysis till now, we have neglected the transaction costs.  In our last test, we include a bid ask-fork of $0.45\%$ in implied volatility terms for every option transaction. More precisely, if we denote by $P^{mid}$ the option price given by the model and by $\sigma^{mid}$ the corresponding implied volatility, the price that our trader must pay to buy an option corresponds to the implied volatility of $\sigma^{mid}+0.225\%$ and the price she receives for selling an option corresponds to the implied volatility of $\sigma^{mid}-0.225\%$. This corresponds roughly to the bid-ask interval observed for the most liquid options on the S\&P 500 index. The true parameters are $\rho=-.9, b=.9$. The misspecified or the market parameters are $\tilde{\rho}=-.1,\tilde{b}=.1$. The evolution of the portfolio performance with 32 rebalancing dates per year using options with 1 month remaining to maturity is plotted in Figure \ref{trans.fig}. With the parameter values that we have chosen, it seems that even in the presence of moderate transaction costs, our strategy produces a positive pay-off in at least $75\%$ of cases. Whether this will be the case in real options markets is a more difficult question, which requires an extensive empirical study. We plan to address this issue in future research.  
\end{itemize}

\begin{figure}
\centerline{\includegraphics[width=0.45\textwidth]{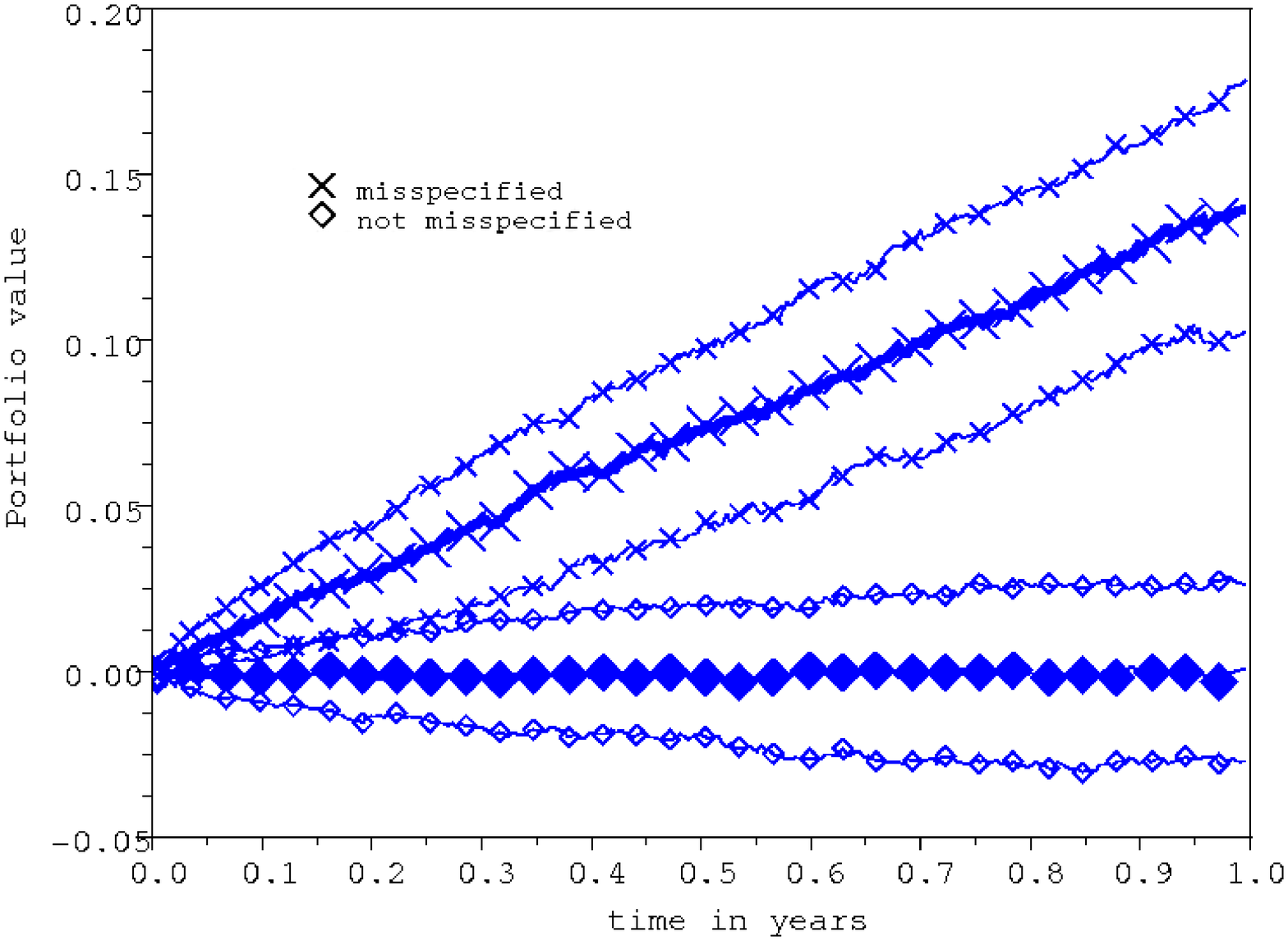}\hspace{-1cm}\includegraphics[width=0.45\textwidth]{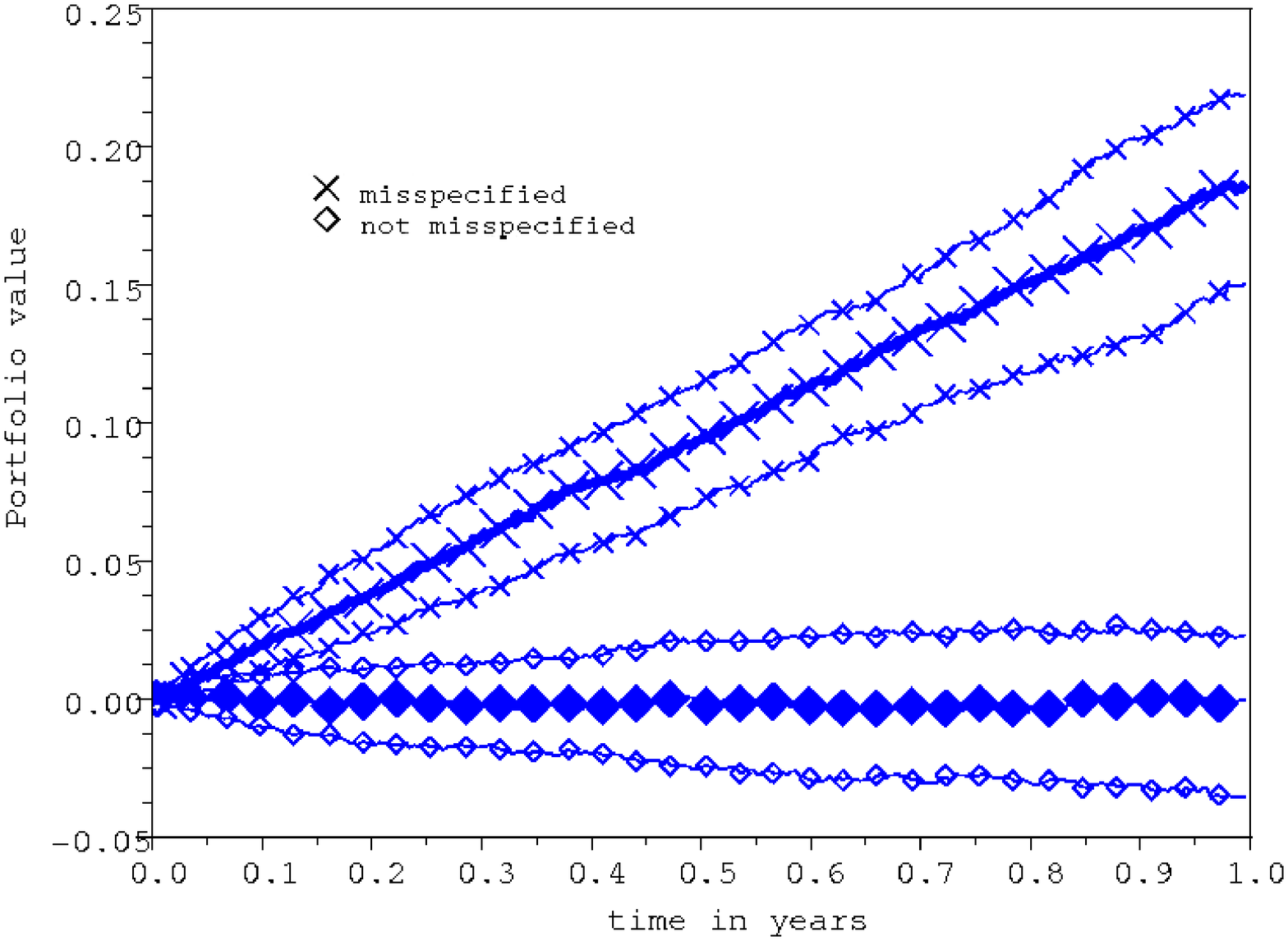}} 
\caption{The evolution of portfolios using options with 1 month to maturity for only correlation misspecification (left) and only volatility of volatility misspecification (right).}
\label{corrvol.fig}
\end{figure}

\begin{figure}
\centerline{\includegraphics[width=0.45\textwidth]{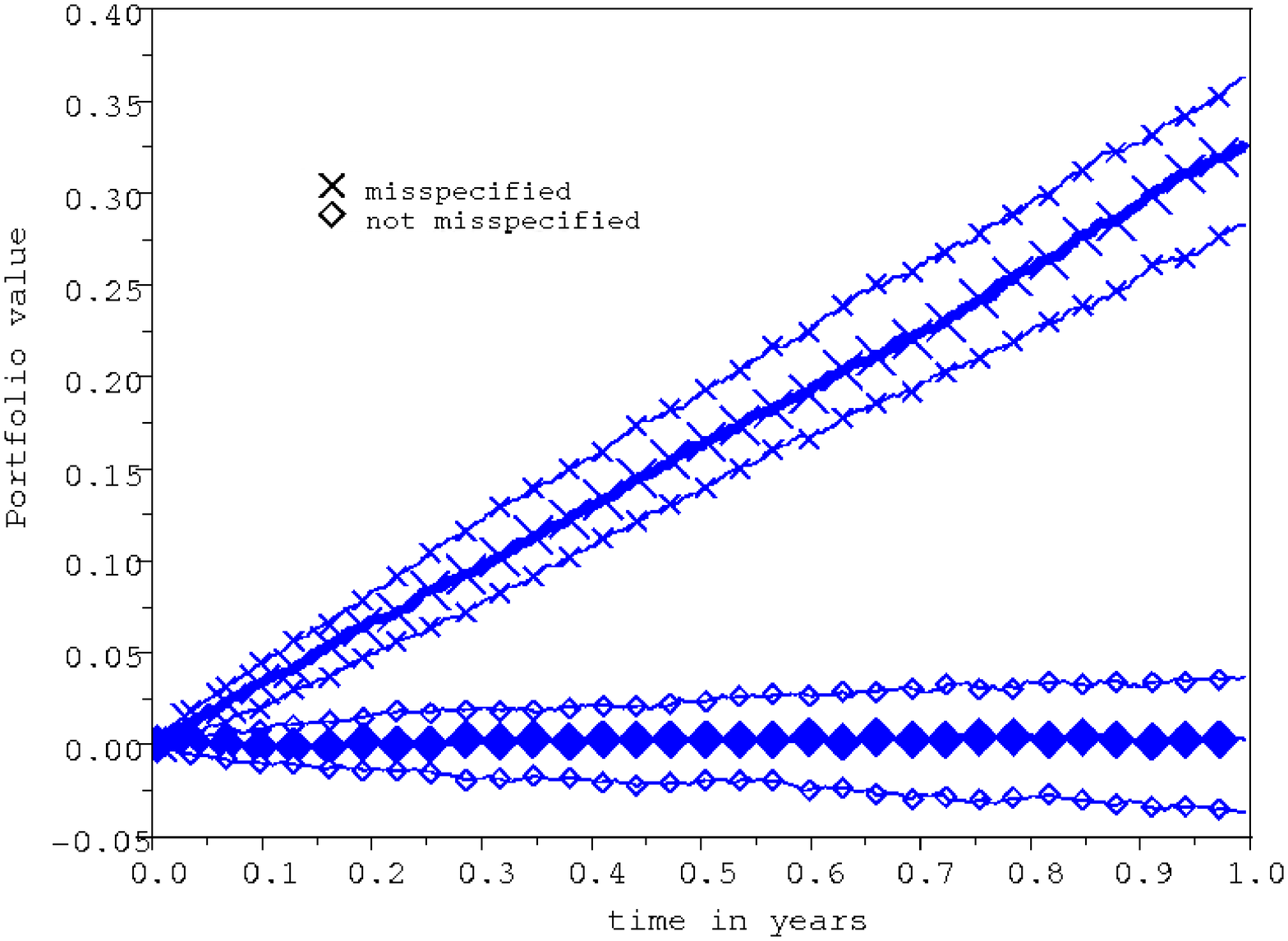}\hspace{-1cm}\includegraphics[width=0.45\textwidth]{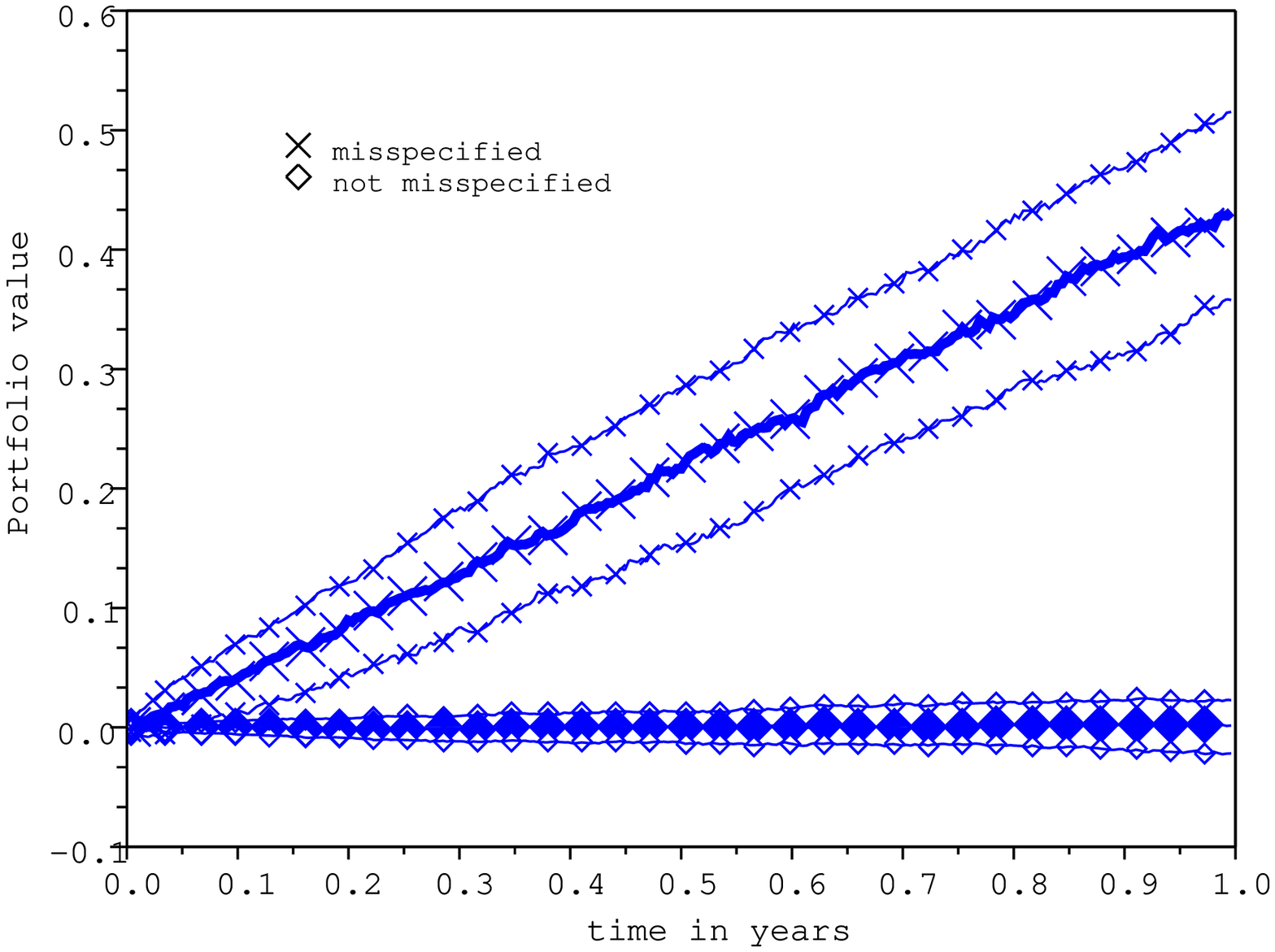}} 
\caption{Impact of the options' maturity: The evolution of portfolios using options with 1 month (left) and 6 months to maturity with $256$ rebalancing dates.}
\label{maturity.fig}
\end{figure}

\begin{figure}
\centerline{\includegraphics[width=0.45\textwidth]{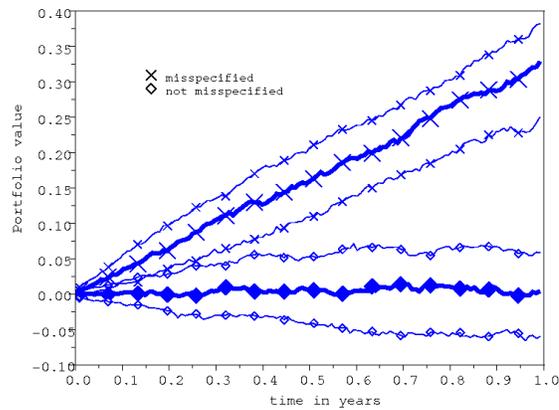}}
\caption{Impact of the rebalancing frequency: Evolution of the portfolio using options with 1 month to maturity and 128 rebalancing dates. Compare with Figure \ref{maturity.fig}, left graph. }
\label{frequency.fig}
\end{figure}

\begin{figure}
\centerline{\includegraphics[width=0.45\textwidth]{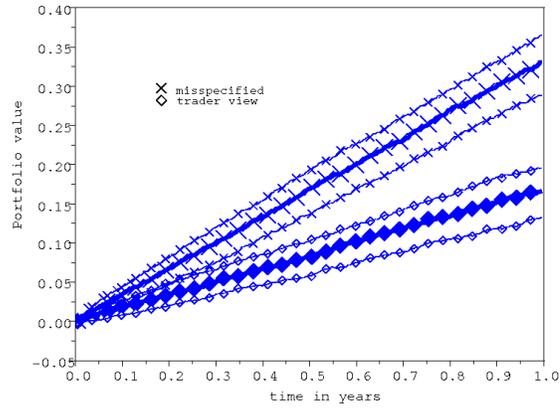}}
\caption{Impact of using a different set of parameters: Evolution of the portfolio for a different set of parameters than the true ones for options with 1 month to maturity and 256 rebalancing dates.}
\label{trader.fig}
\end{figure}

\begin{figure}
\centerline{\includegraphics[width=0.45\textwidth]{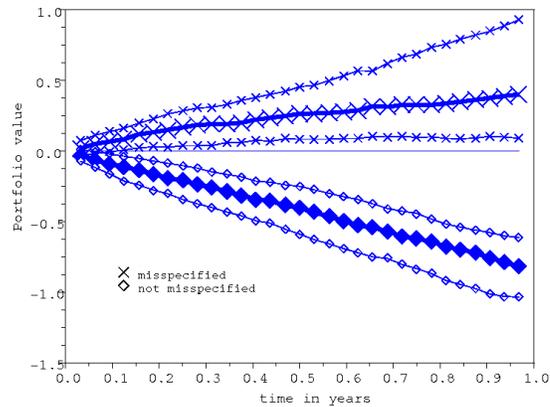}}
\caption{The effect of transaction costs on the performance of the arbitrage strategy.}
\label{trans.fig}
\end{figure}

\section*{Acknowledgements}

We would like to thank the participants of the joint seminar of Soci\'et\'e G\'en\'erale, Ecole Polytechnique and Ecole des Ponts et Chauss\'ees, and especially L. Bergomi (Soci\'et\'e G\'en\'erale) for insightful comments on the previous version of this paper. 

This research of Peter Tankov is supported by the Chair Financial Risks of the Risk Foundation sponsored by Société Générale, the Chair Derivatives of the Future sponsored by the Fédération Bancaire
Française, and the Chair Finance and Sustainable Development sponsored by EDF and Calyon. This research of Rudra Jena is supported by the Chair Financial Risks of the Risk Foundation sponsored by Société Générale.

\appendix
\section{Proof of Lemma \ref{limits.lm}}
First, let us briefly recall some results on fundamental solutions of parabolic PDE \cite{ladyzhenskaya.al.68}.
Let \[
L\left(x,t,\frac{\partial}{\partial x},\frac{\partial}{\partial t}\right)u:=\frac{\partial}{\partial t} + \frac{1}{2}\sum_{i,j=1}^{n}a_{ij}(x,t)\frac{\partial^{2}u}{\partial x_{i}\partial x_{j}}+\sum_{i=1}^{n}b_{i}(x,t)\frac{\partial u}{\partial x_{i}} + c(x,t).\] 
\begin{assu}\label{as1.ass}
 There is a positive constant $\mu$ such that
\[
\sum_{i,j=1}^n a_{ij}(x,t)\xi_{i}\xi_{j}\geqslant\mu|\xi|^{2}\quad \forall (x,t)\in \mathbb R^{n}\times[0,T],\:\xi\in \mathbb R^{n}.
\]
\end{assu}
\begin{assu}\label{as2.ass}
There exists $\alpha\in (0,1)$ such that the coefficients of $L$ are bounded and Hölder continuious in $x$ with exponent $\alpha$ and Hölder continuous in $t$ with exponent $\frac{\alpha}{2}$, uniformly
with respect to $(x,t)$ in $\mathbb R^{n}\times[0,T]$.
\end{assu}

The \emph{fundamental solution} of the parabolic second-order equation with operator $L$ is the function $\Gamma(x,t,\xi,T)$ which satisfies
$$
L\left(x,t,\frac{\partial}{\partial x},\frac{\partial}{\partial t}\right) \Gamma(x,t,\xi,T) = \delta(x-\xi)\delta(t-T),\quad t\leq T. 
$$
Under \ref{as1.ass} and \ref{as2.ass}, the operator $L$ admits a fundamental solution $\Gamma$ with
\begin{align}
|D^r_t D^s_x \Gamma(x,t,\xi,T)| \leq c(T-t)^{-\frac{n+2r+s}{2}}\exp\left(-C\frac{|x-\xi|}{T-t}\right)\label{fundsolbound}
\end{align}
where $r$ and $s$ are integers with $2r+s\leq 2$, $t<T$ and $c,C$ are positive. 

Consider now the Cauchy problem,
\begin{align*}
Lu(x,t)&=f(x,t),\quad (x,t)\in \mathbb R^{n} \times [0,T),    \\
u(x,T) &=\phi(x),\quad x\in \mathbb R^{n},
\end{align*}
where $f$ is Hölder continuous in its arguments, $\phi$ is continuous and these function satisfy reasonable growth constraints at infinity. Then the solution to this problem can be written as
$$
u(x,t) = \int_t^T d\tau \int_{\mathbb R^n} d\xi \Gamma(x,t,\xi,\tau) f(\xi,\tau) + \int_{\mathbb R^n} d\xi \Gamma(x,t,\xi,T)\phi(\xi).
$$

Let us now turn to the proof of Lemma \ref{limits.lm}. We discuss the results for the put options. The result for calls follows directly by put-call parity. 
As described in Section \ref{setting.sec}, the put option price $P(S,y,t)$ solves the PDE \eqref{pde}. Let $x = \log\frac{S}{K}$ and $p(x,y,t):=P(Ke^x,y,t)$. Then $p$ solves the PDE
\begin{align}
\frac{\partial p}{\partial t} + \mathcal A p = 0,\qquad p(x,y,T)= & K(1-e^x)^+.
\label{xpde}
\end{align}
where 
\begin{align*} 
\mathcal{A}p= & \frac{1}{2}\sigma^2 \left\{\frac{\partial^2p}{\partial x^2} - \frac{\partial p}{\partial x}\right\}+a\frac{\partial p}{\partial y}+\frac{1}{2}b^2\frac{\partial^2p}{\partial y^2}+\rho b\sigma \frac{\partial^2p}{\partial x \partial y} \label{mainop}
\end{align*} 
The quantities $a,b,\sigma,\rho$ correspond to market misspecified values, but in this appendix we shall omit the tilda and often drop the explicit dependence on $t$ and $y$ in model parameters to simplify notation. 

Therefore, the option price can be written as 
\begin{align}
p(x,y,t) = \int dz \int dv\, \Gamma(x,y,t,z,v,T)K(1-e^z)^+,
\end{align}
where $\Gamma(x,y,t,z,v,T)$ is the fundamental solution of \eqref{xpde}. Since the coefficients of $\mathcal A$ do not depend on $x$, $\Gamma(x,y,t,z,v,T) \equiv  \Gamma(0,y,t,z-x,v,T)$. Coming back to the original variable $S$, we have
\begin{align}
P(S,y,t) = \int dz \, \gamma(y,t,z-\log\frac{S}{K},T)K(1-e^z)^+\label{fundsolput}
\end{align}
where we write
$$
\gamma(y,t,z,T):=\int dv\, \Gamma(0,y,t,z,v,T).
$$

\paragraph{Decay of the gamma $\frac{\partial^2 P}{\partial S^2}$}
By direct differentiation of \eqref{fundsolput}, we get
\begin{align}
\frac{\partial^2 P}{\partial S^2} = \frac{K}{S^2} \gamma(y,t,\log \frac{K}{S},T)\label{gammarep}
\end{align}
and it follows from \eqref{fundsolbound} that
\begin{align}
\left|\frac{\partial^2 P}{\partial S^2}\right| \leq \frac{C K}{S^2|T-t|^{\frac{1}{2}}}\exp\left(-\frac{c\log^2 \frac{K}{S}}{T-t}\right).\label{gammabound} 
\end{align}
This proves the required decay properties, and the continuity of $\frac{\partial^2 P}{\partial S^2}$ also follows from \eqref{fundsolbound}. 

\paragraph{Decay of the vega $\frac{\partial P}{\partial y}$} We denote $U(S,y,t):=\frac{\partial P (S,y,t)}{\partial y}$ and $u(x,y,t):=\frac{\partial p (x,y,t)}{\partial y}$. Using the regularity of coefficients and local regularity results for solutions of parabolic PDE \cite[Corollary 2.4.1]{krylov.08}, we conclude that the derivatives $\frac{\partial^3 p}{\partial x^2 \partial y}$, $\frac{\partial^3 p}{\partial x \partial y^2}$ and $\frac{\partial^3 p}{\partial y^3}$ exist, and therefore the operator \eqref{mainop} may be differentiated term by term with respect to $y$, producing
\begin{align}\label{vegaeqapp}
\frac{\partial u}{\partial t}+\mathcal{A}_{1}u=-\sigma \sigma'  \left\{\frac{\partial^2 p}{\partial x^2} - \frac{\partial p}{\partial x}\right\},
\end{align}
where
\begin{align*}
 \mathcal{A}_{1}= & \mathcal{A}+(\rho b \sigma)' \frac{\partial }{\partial x}+ bb^{'}\frac{\partial }{\partial y} + a'.
\end{align*}
All the primes denote the derivative w.r.t. $y$ and the terminal condition is $u(S,y,T)\equiv 0$, since the original terminal condition is independent of $y$.

The right-hand side of \eqref{vegaeqapp} satisfies
\begin{align}
\frac{\partial^2 p}{\partial x^2} - \frac{\partial p}{\partial x} = K\gamma(y,t,-x,T)\label{2derrep}
\end{align}
so from \eqref{fundsolbound} we get
\begin{align}
\left|\frac{\partial^2 p}{\partial x^2}-\frac{\partial p}{\partial x}\right| \leq \frac{C K}{(T-t)^{\frac{1}{2}}}\exp\left(-\frac{c x^2}{T-t}\right).\label{xgammabound} 
\end{align}

Let $\Gamma_1$ denote the fundamental solution of the parabolic equation with the operator appearing in the left-hand side of \eqref{vegaeqapp}. Using the estimates of the fundamental solution in \cite[section 4.13]{ladyzhenskaya.al.68} (in particular, the Hölder continuity) and the bound \eqref{xgammabound}, we can show that the solution to \eqref{vegaeqapp} is given by
$$
u(x,y,t) = \int_t^T dr \int_{\mathbb R^2} dz\,dv\, \Gamma_1 (x,y,t,z,v,r) \sigma(v)\sigma'(v)   \left\{\frac{\partial p}{\partial z} - \frac{\partial^2 p}{\partial z^2}\right\}(z,v,r).
$$
Using the boundedness of $\sigma$ and $\sigma'$, the bound on the fundamental solution and \eqref{xgammabound}, and integrating out the variable $v$, we get
$$
|u(x,y,t)| \leq \int_t^T dr \int dz  \frac{C K}{(T-r)^{\frac{1}{2}}(r-t)^{\frac{1}{2}}}\exp\left(-\frac{cz^2}{T-r}-\frac{c(x-z)^2}{r-t}\right) .
$$
Explicit evaluation of this integral then yields the bound
\begin{align}
|u(x,y,t)|&\leq C K \sqrt{T-t} e^{-\frac{cx^2}{T-t}}\label{vegabound} \\
|U(S,y,t)|&\leq C K \sqrt{T-t} e^{-\frac{c\log^2 \frac{S}{K}}{T-t}}
\end{align} 
from which the desired decay properties follow directly.

\paragraph{Decay of $\frac{\partial^2 P}{\partial S\partial y}$ } We denote $w(x,y,t) = \frac{\partial u}{\partial x}$ and $W(S,y,t)=\frac{\partial U}{\partial S} = \frac{w}{S}$ and differentiate equation \eqref{vegaeqapp} with respect to $x$:
\begin{align*}
\frac{\partial w}{\partial t}+ \mathcal{A}_{1}w = -\sigma \sigma'  \left\{\frac{\partial^3 p}{\partial x^3} - \frac{\partial^2 p}{\partial x^3}\right\}.
\end{align*}
From \eqref{2derrep} and \eqref{fundsolbound},
\begin{align}
\left|\frac{\partial^3 p}{\partial x^3}-\frac{\partial^2 p}{\partial x^2}\right| \leq \frac{C K}{T-t}\exp\left(-\frac{c x^2}{T-t}\right). 
\end{align}
Similarly to the previous step we now get:
$$
|w(x,y,t)| \leq \int_t^T dr \int dz  \frac{C K}{(T-r)(r-t)^{\frac{1}{2}}}\exp\left(-\frac{cz^2}{T-r}-\frac{c(x-z)^2}{r-t}\right)
$$
and explicit evaluation of this integral yields the bounds
\begin{align}
|w(x,y,t)|&\leq C K e^{-\frac{cx^2}{T-t}},\qquad |W(S,y,t)|\leq C \frac{K}{S} e^{-\frac{c\log^2 \frac{S}{K}}{T-t}}\label{wbound} 
\end{align}

\paragraph{Decay of $\frac{\partial^2 P}{\partial y^2}$} We denote $v(x,y,t)=\frac{\partial u}{\partial y}$, $V(S,y,t) = \frac{\partial U}{\partial y}$ and differentiate equation \eqref{vegaeqapp} with respect to $y$:
\begin{align*}
\frac{\partial v}{\partial t}+ \mathcal{A}_{2}v=- \sigma\sigma'  \frac{\partial}{\partial y}\left\{\frac{\partial^2 p}{\partial x^2}-\frac{\partial p}{\partial x}\right\}-a''u-(\rho b \sigma )^{''}w
\end{align*}
where
\begin{align*}
 \mathcal{A}_2= & \mathcal{A}_1+bb^{'}\frac{\partial}{\partial y}+(\rho b \sigma)^{'}\frac{\partial}{\partial x} + a' + (bb')'.
\end{align*}
Once again, from \eqref{2derrep} and \eqref{fundsolbound},
\begin{align*}
\left|\frac{\partial}{\partial y}\left\{\frac{\partial^2 p}{\partial x^2}-\frac{\partial p}{\partial x}\right\}\right| \leq \frac{C K}{T-t}\exp\left(-\frac{c x^2}{T-t}\right).
\end{align*}
Using this bound together with \eqref{wbound} and \eqref{vegabound} and proceding as in previous steps, we complete the proof.

\end{document}